%% file: smooth_major.tex
\documentclass[superscriptaddress,aps,pra,amssymb,amsfonts,letterpaper]{revtex4}
\pdfoutput=1
\usepackage{hyperref}
\usepackage{enumerate}
\usepackage{verbatim}
\usepackage{amsmath}
\usepackage{latexsym}
\usepackage{revsymb}
\usepackage{yfonts}
\usepackage{amsfonts}
\usepackage{amsmath}
\usepackage{amssymb}
\usepackage{amsthm}
\usepackage{graphicx}
\usepackage{bm}
\usepackage{bbm}
\usepackage{color}
\usepackage{mathrsfs}
\usepackage{epstopdf}
\input{./definitions.tex}
\begin{document}
\title{Approximate Majorization}
\author{Micha\l\ \surname{Horodecki}}
\affiliation{\small Institute of Theoretical Physics and Astrophysics, National Quantum Information Centre, Faculty of Mathematics, Physics and Informatics, Univeristy of Gda{\'n}sk, Wita Stwosza 57, 80-308 Gda{\'n}sk, Poland}
\author{Jonathan \surname{Oppenheim}}
\affiliation{University College of London, Department of Physics \& Astronomy, London, WC1E 6BT and London Interdisciplinary Network for Quantum Science}                        
\author{Carlo \surname{Sparaciari}}
\affiliation{University College of London, Department of Physics \& Astronomy, London, WC1E 6BT}
\begin{abstract}
Although an input distribution may not majorize a target distribution, it may majorize a distribution which
is close to the target. Here we introduce a notion of approximate majorization. For any distribution, and
given a distance $\delta$, we find the approximate distributions which majorize (are majorized by) all other
distributions within the distance $\delta$.  We call these the steepest and flattest approximation. This enables
one to compute how close one can get to a given target distribution under a process governed by majorization.
We show that the flattest and steepest approximations preserve ordering under majorization. Furthermore,
we give a notion of majorization distance. This has applications ranging from thermodynamics, entanglement
theory, and economics.
\end{abstract}
\date{\today}
\maketitle
The theory of majorization~\cite{book-majorization, nielson-majorisation} has important applications in topics
as diverse as matrix theory, geometry, combinatorics, statistics, thermodynamics, entanglement theory and
economics. It defines a partial ordering over vectors of real numbers, as follow. For two vectors $a, b \in \R^k$,
one define $a^{\downarrow}$, $b^{\downarrow}$ as the same vectors whose elements are non-increasingly
ordered. Then, one says that \emph{$a$ weakly majorizes $b$ from below},
$a \succ_{w} b$, iff
\be
\sum_{i=1}^l a^{\downarrow}_i \geq \sum_{i=1}^l b^{\downarrow}_i, \quad  \forall \, l = 1 , \ldots , k.
\ee
When the two vectors have the same norm, one says that $a$ \emph{majorizes} $b$, $a \succ b$.
Hardy, Littlewood, and Polya~\cite{Hardy-Littlewood-Polya} showed that $a \succ b$ iff $b=Da$,
where $D$ is a doubly-stochastic matrix (alternatively a probabilistic mixture of permutations).
\par
Here, we are interested in normalised vectors $p$ and $q$ which represent an input and output
probability distribution of $k$ elements. In many situations, processes on these systems are represented
by doubly-stochastic matrices. Majorization then determines whether there exists a process which
takes $p$ to $q$, but in many situations, we are more interested in whether a process gets us close
to the target distribution. In the context of single-shot information theory, and of certain entropic functions,
finding an approximation to the target distribution which minimises resources has been termed
{\it smoothing}~\cite{Renner-PhD}. Here we are interested in a different notion of smoothing which can be
applied to finding the optimal approximation of the output or input state for the purposes of majorization.
The present work, whose initial draft circulated in 2013, has recently found application in the context
of thermodynamics~\cite{deMeer-MSc,deMeer2017}, and the smoothing we use has
been independently rediscovered in the context of convex optimisation~\cite{nila-2017}. 
\par
Here, we investigate how majorization behave under smoothing. 
We first introduce two smoothed versions of a given probability distribution, namely, the
steepest and flattest $\delta$-approximation of this distribution. Then, we show that the
steepest approximation majorizes any probability distribution whose distance from the original
distribution is less or equal than $\delta$, while the flattest approximation is majorized by all
these distributions. We also show that smoothing preserves monotonicity under majorization,
for both the steepest and flattest approximation. Finally, we apply our findings to the analysis
of the smooth version of Schur concave/convex functions.
As we anticipated, our main tool consists in two specific approximations of a given probability
distribution $p$, each of them $\delta$-close to the original distribution. These approximations
are (i) the {\it \flattest\ $\delta$-approximation} of $p$, and (ii) the {\it \stippest\ $\delta$-approximation}
of $p$. In the following we will assume the elements of the probability distribution $p$ to be
non-increasingly ordered.
\par
The \stippest\ $\delta$-approximation of $p$, which we denote by $\stipp$, is constructed as follows.
If $\| p - \e_1 \| \leq \delta$, where $\e_1$ is the distribution whose first element is equal to 1, then we take
$\stipp = \e_1$. Otherwise, we maximally increase the largest element of $p$, and we cut the tail. More
precisely, we first add $\frac{\delta}{2}$ to the largest element of $p$ (which is possible, since
$\| p - \e_1 \| > \delta$). This procedure returns a non-normalized distribution which we will denote
by $r$, whose elements are defined as
\be
\label{eq:nnormr}
r_i=
\left\{
\bea{lll}
p_1+\frac{\delta}{2} &\text{for}& i=1, \\
p_i & \text{for}&i\not=1. \\
\eea
\right.
\ee
Then we cut $\frac{\delta}{2}$ from the tail of this distribution. Formally, we take the integer $l^* \in
\left\{ 1, \ldots, k \right\}$ such that 
\be
\label{eq:deflstar}
\sum_{i=1}^{l^*} r_i \leq 1 \quad \text{and} \quad \sum_{i=1}^{l^*+1} r_i > 1,
\ee
and we define the \stippest\ $\delta$-approximation of $p$ as
\be
\stipp_i=
\left\{
\bea{lll}
r_i & \text{for} & i< l^*+1, \\
1 - x & \text{for} & i=l^*+1, \\
0 & \text{for} & i>l^*+1. \\
\eea
\right.
\ee
where $x = \sum_{i=1}^{l^*} r_i$. In Fig.~\ref{fig:flatstip}, the process of steepening a probability
distribution $p$ is shown, together with the resulting \stippest\ $\delta$-approximation $\stipp$.
\par
The \flattest\ $\delta$-approximation of $p$, denoted by $\flatp$, is constructed in the following way.
If $\|p-\eta\|\leq \delta$, where $\eta $ is the uniform distribution, then we define $\flatp=\eta$. 
Otherwise, we proceed as follows. For a given $x, y \in \left[0,1\right]$, we define the following subsets,
\begin{align}
\label{eq:setI}
I &= \left\{ i \in \left\{ 1, \ldots, k \right\} | \ p_i \geq x \right\}, \\
\label{eq:setJ}
J &= \left\{ i \in \left\{ 1, \ldots, k \right\} | \ p_i \leq y \right\},
\end{align}
and we introduce the functions
\begin{align}
\label{eq:epsx}
\ep(x) &= \sum_{i \in I} \, ( p_i - x ), \\
\label{eq:gammax}
\gamma(y) &= \sum_{i \in J} \, ( y - p_i ).
\end{align}
Then, we choose $x^* \in [0,1]$ such that $\ep(x^*)=\frac{\delta}{2}$, and $y* \in [0,1]$ such
that $\gamma(y^*) = \frac{\delta}{2}$. It is worth noting that, since $\| p - \eta \| > \delta$, both
$x^*$ and $y^*$ exist and are unique, and moreover $x^* > y^*$.
We can now define the \flattest\ $\delta$-approximation of $p$ as
\be
\label{eq:flatp}
\flatp_i=
\left\{
\bea{lll}
 x^* & \text{for} & i\in I,\\
 y^* & \text{for} & i\in J,\\
p_i &\text{else}.& \\
\eea
\right.
\ee
In Fig.~\ref{fig:flatstip}, the process of flattening a probability distribution $p$ is shown,
together with the resulting \flattest\ $\delta$-approximation $\flatp$.
\par
{\bf Remark.} Let us note that the above constructions preserve the order of the elements, i.e., if the
probability distribution $p$ is non-increasingly ordered, then the same applies to both $\stipp$
and $\flatp$.
\begin{figure}[ht!]
  \center
  \includegraphics[width=0.185\textwidth]{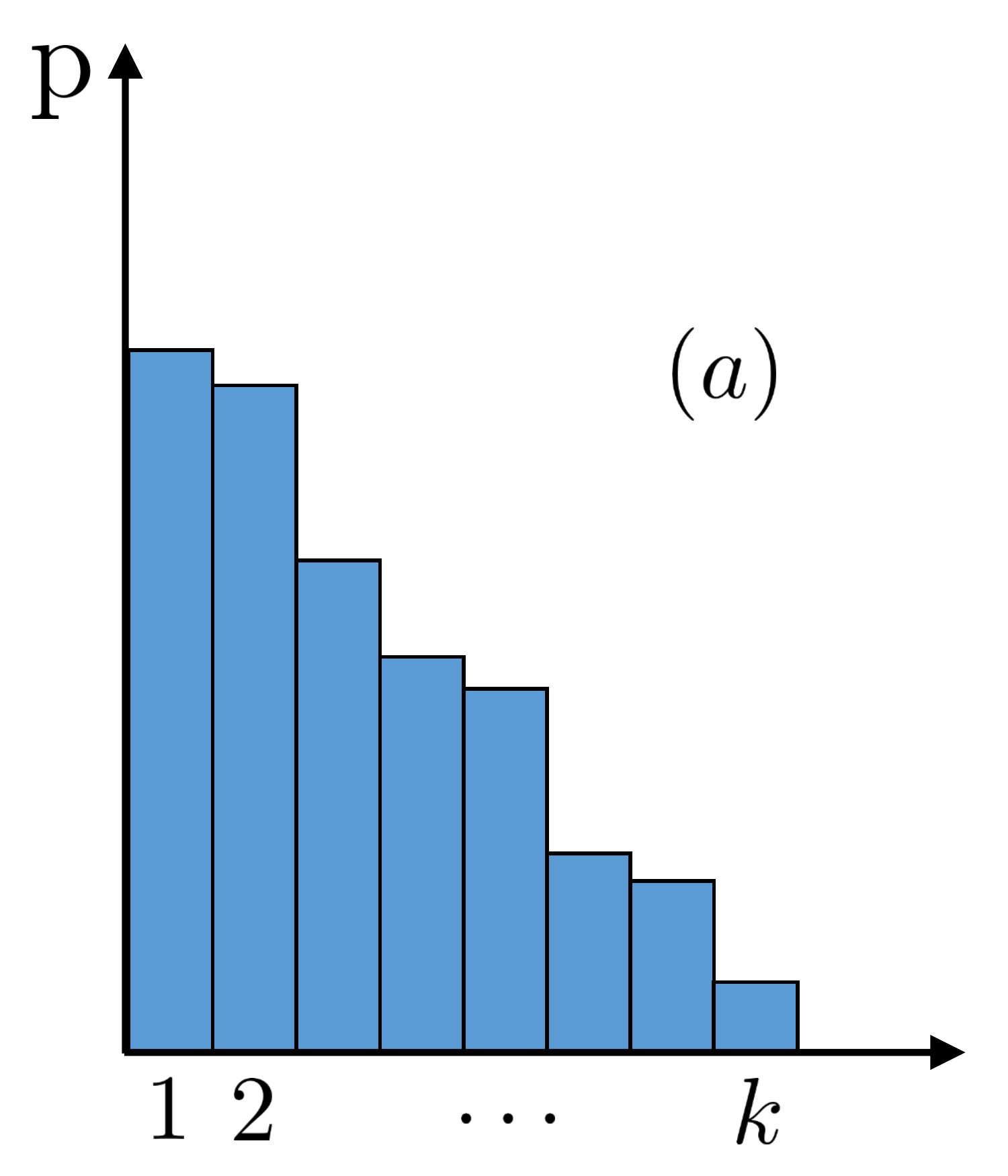}
  \includegraphics[width=0.185\textwidth]{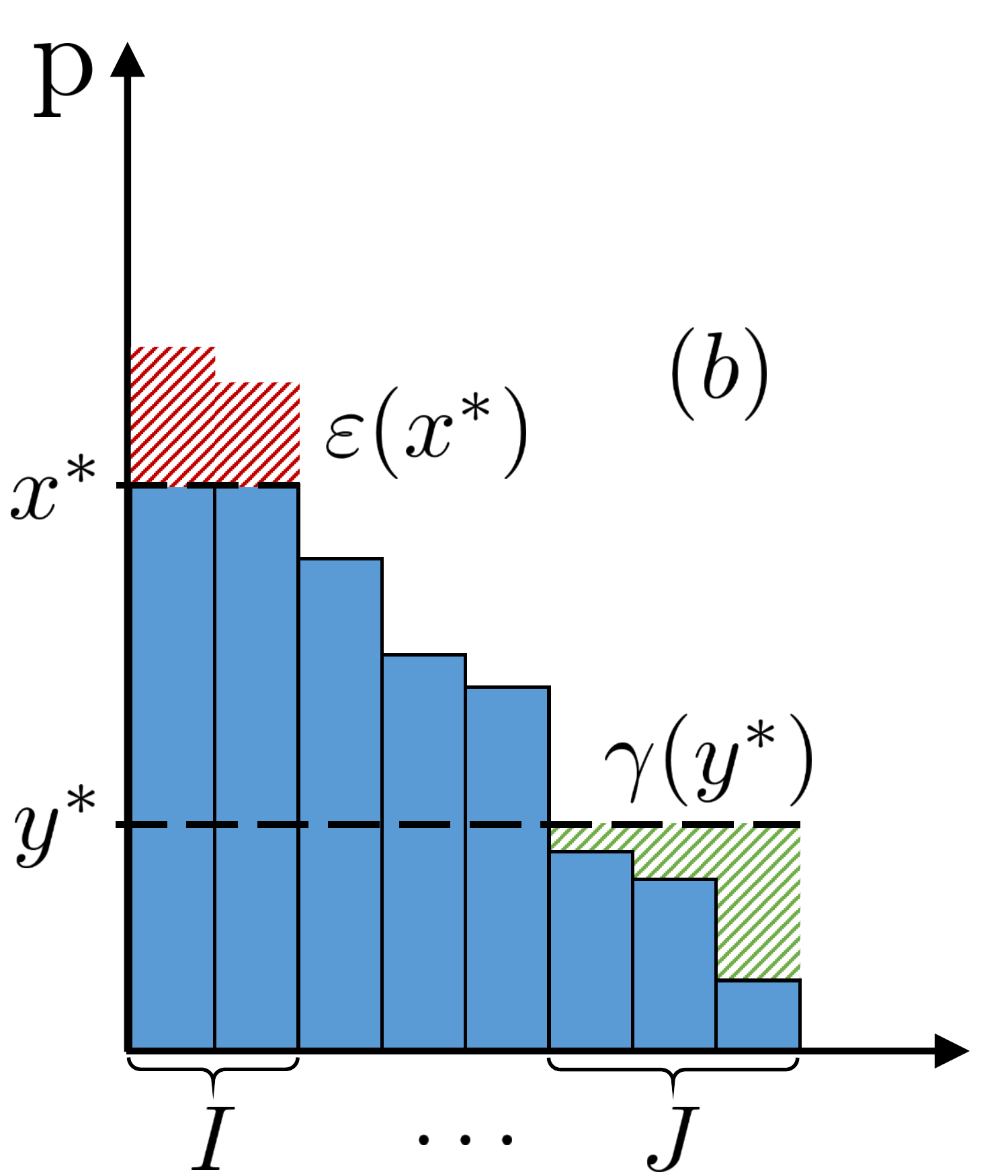}
  \includegraphics[width=0.2\textwidth]{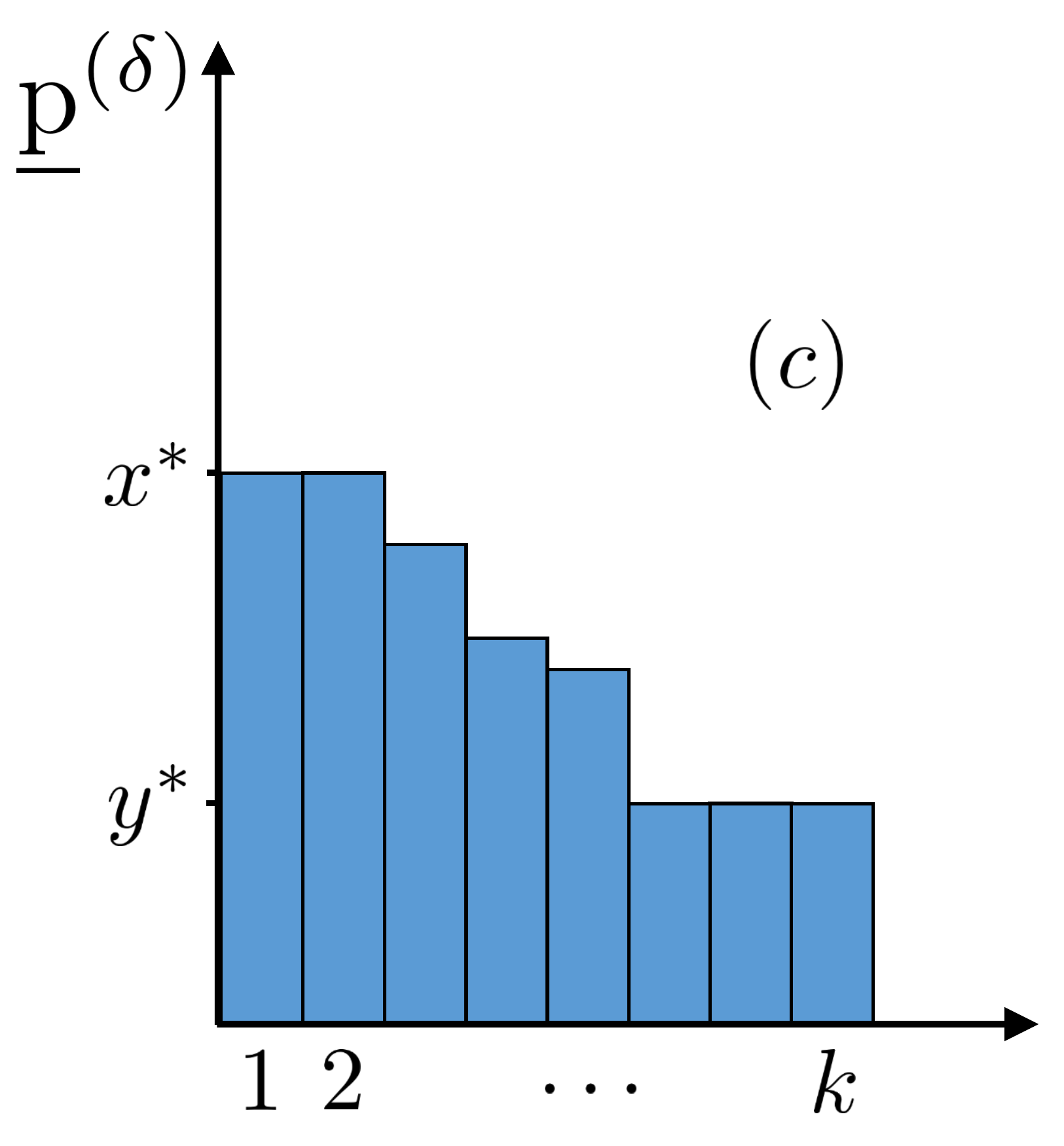}
  \includegraphics[width=0.185\textwidth]{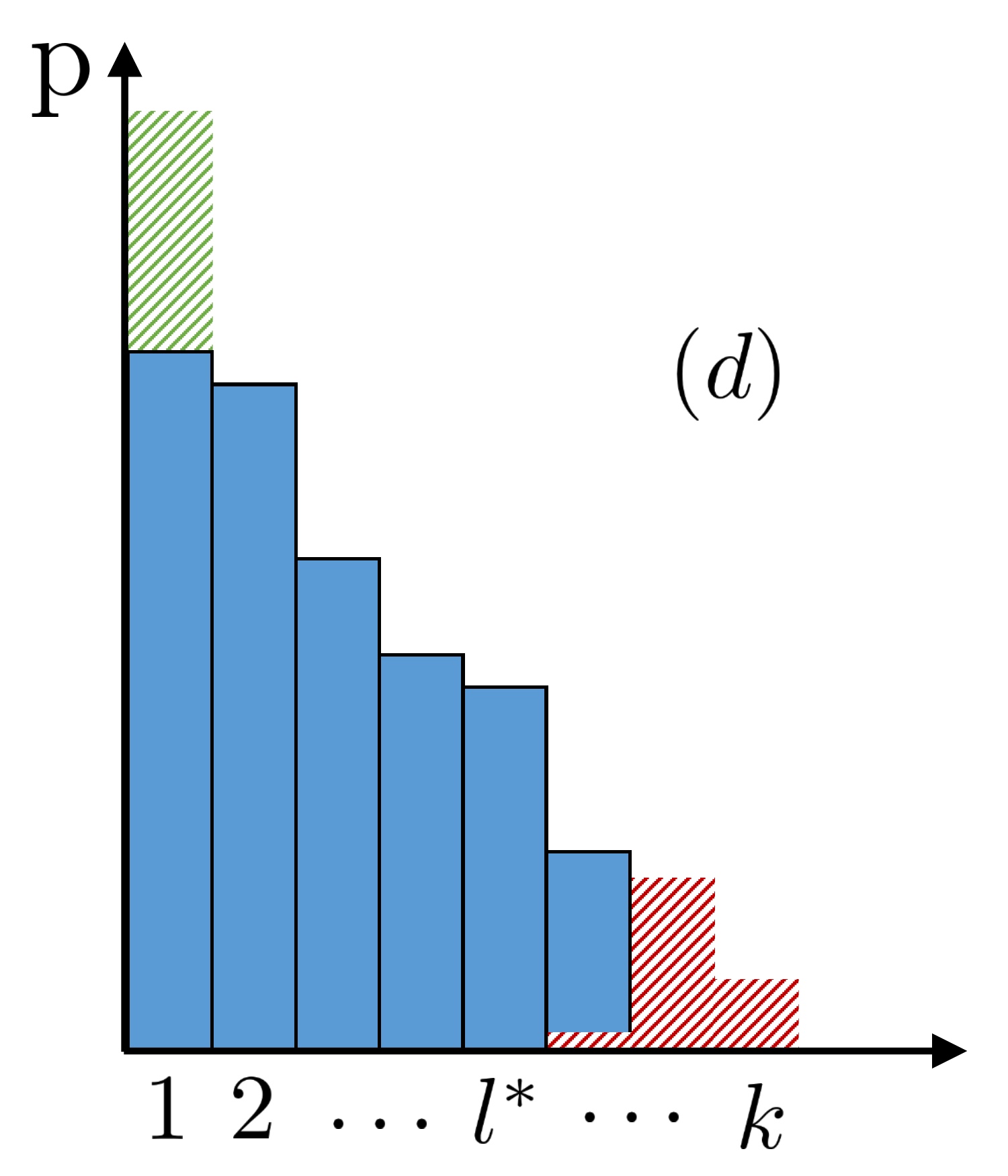}
  \includegraphics[width=0.2\textwidth]{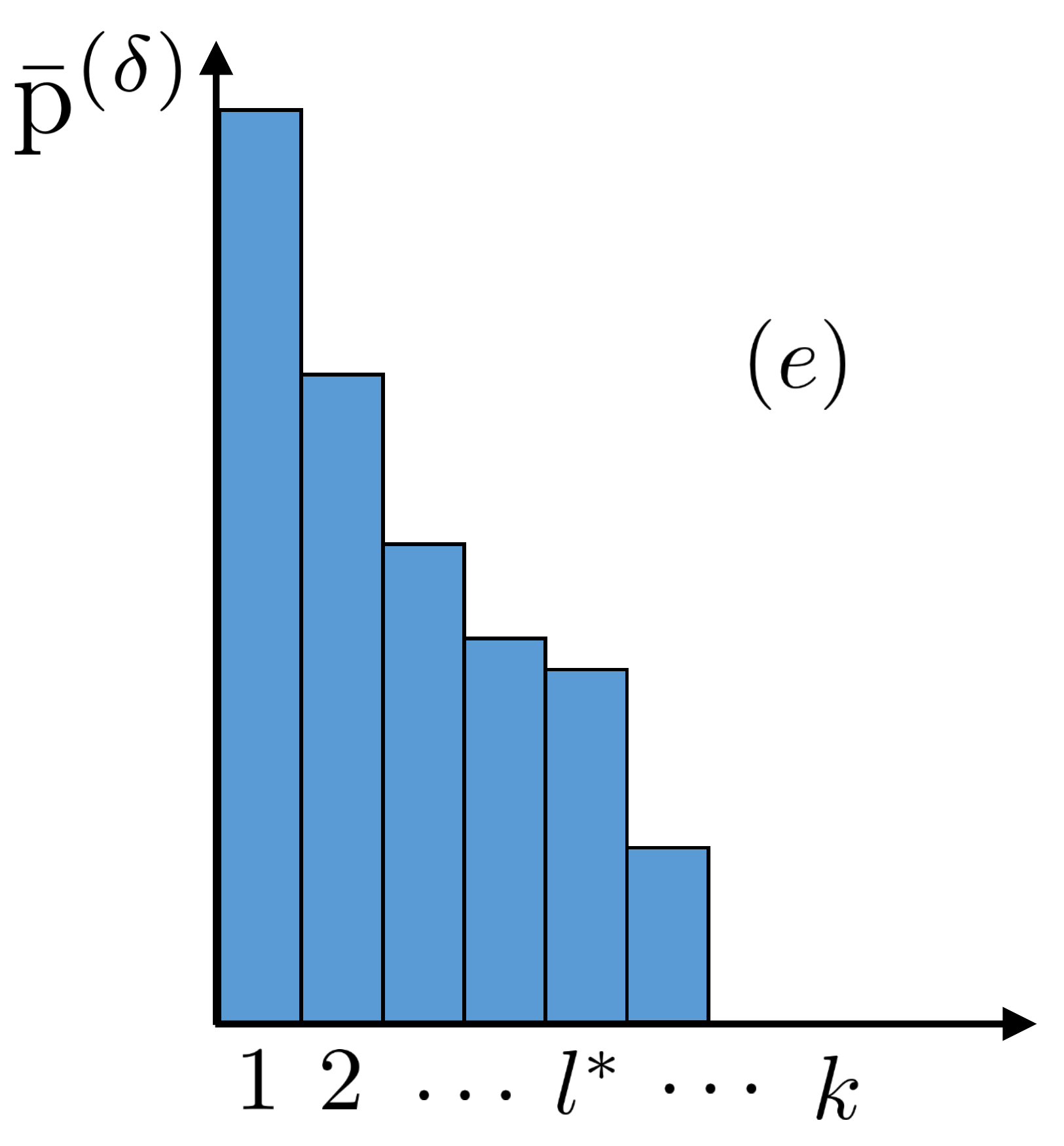}
  \caption{The procedure of flattening and steepening the probability distribution $p$.
  The added portion is green, while the removed one is red. The two portions have the same
  area equal to $\frac{\delta}{2}$. (a) The original distribution $p$. (b) The procedure of flattening
  the probability distribution $p$. (c) The flattest $\delta$-approximation of $p$, $\flatp$. (d) The
  procedure of steepening the distribution $p$. (e) The steepest $\delta$-approximation of $p$,
  $\stipp$.}
  \label{fig:flatstip}
\end{figure}
\par
The following lemma, concerning the majorization properties of $\stipp$ and $\flatp$, singles
out these two distributions among all the other distributions which are $\delta$-close to $p$.
\def\flatstiplemma
{
For a given probability distribution $p$ of $k$ elements, the distributions $\stipp$ and $\flatp$ are 
extremal $\delta$-approximations of $p$ in the sense of majorization order,
\bei
\item[(i)] The \stippest\ $\delta$-approximation $\stipp$ of $p$
majorizes any arbitrary distribution $p'$ satisfying $\|p-p'\|\leq \delta$,
\be
\stipp\succ p'.
\ee 
\item[(ii)] The \flattest\ $\delta$-approximation $\flatp$ of $p$ is majorized by every arbitrary distribution
$p'$ satisfying $\|p-p'\|\leq \delta$,
\be
p' \succ \flatp.
\ee
\eei
}
\begin{lemma}
\flatstiplemma
\label{lem:flatstiplemma}
\end{lemma}
Thus, the above lemma shows that the \stippest\ and \flattest\ $\delta$-approximations of $p$
are extremal points (with respect to the majorization order) of the set of all probability distributions
which are $\delta$-close to $p$. We can also consider the Lorenz curves associated with the
distributions $p$, $\stipp$, and $\flatp$, see Fig.~\ref{fig:flatstip_lorenz}. The Lorenz curve of 
$\stipp$ is obtained by shifting all elbows upward by $\frac{\delta}{2}$, until we reach the
normalisation threshold equal to 1. Then, the curve is concluded by an horizontal segment.
Formally, the Lorenz curve of the \stippest\ $\delta$-approximation is defined as
\be
L_{\stipp}(l) = \sum_{i=1}^l \stipp_i =
\left\{
\bea{lll}
 \sum_{i=1}^l p_i + \frac{\delta}{2} & \text{for} & l \leq l^*,\\
 1 & \text{for} & l > l^*.\\
\eea
\right.
\ee
The Lorenz curve of $\flatp$ begins as a straight segment connecting the origin of the axes
with the point $\left( l_I , \sum_{i \in I} p_i - \frac{\delta}{2} \right)$, where $l_I$ is the maximum
index of the set $I$. The final part of the curve is also a straight segment, connecting
the point $\left( l_J - 1 , 1 - \left( \sum_{i \in J} p_i + \frac{\delta}{2} \right) \right)$, where
$l_J$ is the minimum index of the set $J$, with the point $\left( k, 1 \right)$. Finally, the
other elbows of the curve are simply shifted downward by $\frac{\delta}{2}$. More formally,
the Lorenz curve of the \flattest\ $\delta$-approximation is
\be
L_{\flatp}(l) = \sum_{i=1}^l \flatp_i =
\left\{
\bea{lll}
 l \, x^* & \text{for} & l \in I,\\
 \sum_{i=1}^l p_i - \frac{\delta}{2} & \text{for} & l \not\in I \cup J,\\
 1 - \left(k - l\right) y^* & \text{for} & l \in J.\\
\eea
\right.
\ee
Then, from Lemma~\ref{lem:flatstiplemma} it follows that the Lorenz curve of any probability distribution $p'$
(such that $\| p - p' \| \leq \delta$) entirely lies above the Lorenz curve of $\flatp$, and below the
Lorenz curve of $\stipp$.
\begin{figure}[ht!]
  \center
  \includegraphics[width=0.3\textwidth]{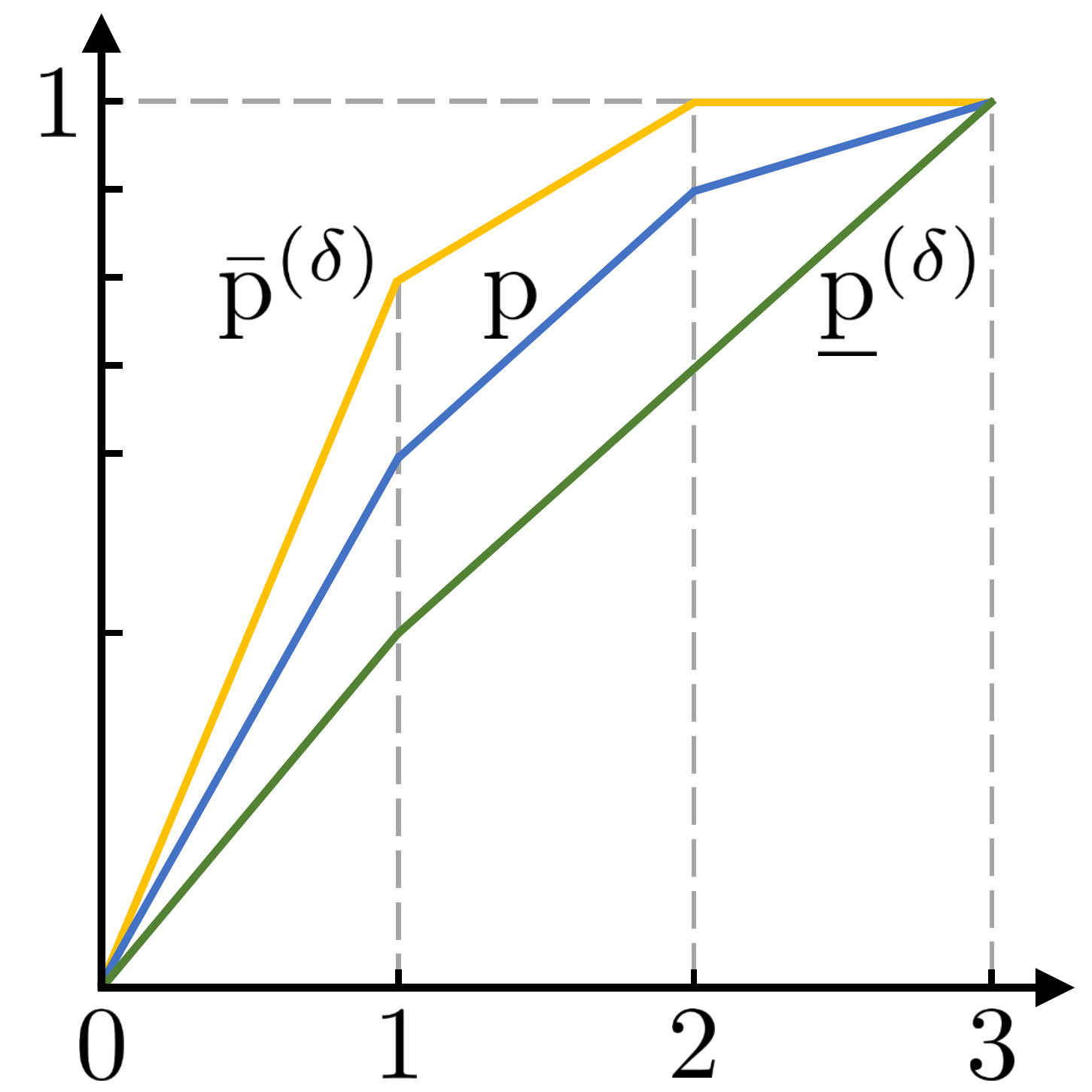}
  \caption{The Lorenz curve of the probability distribution
  $p = \left( 0.6, 0.3, 0.1 \right)$ is shown in blue. For $\delta = 0.4$, we find
  that the \stippest\ $\delta$-approximation of $p$ is $\stipp = \left( 0.8, 0.2, 0 \right)$,
  and its Lorenz curve is shown in orange. The \flattest\ $\delta$-approximation of $p$
  is $\flatp = \left( 0.4, 0.3, 0.3 \right)$, and its Lorenz curve is shown in green.}
  \label{fig:flatstip_lorenz}
\end{figure}
\par
An additional property of the processes of steepening and flattening a probability
distribution consists in the fact that they preserve the majorization order.
\def\preservingorderlemma
{
Given two probability distributions of $k$ elements, $p$ and $q$, which satisfy $p \succ q$,
we have 
\begin{align}
\flatp \succ \flatq, \\
\stipp \succ \stipq.
\end{align}
}
\begin{lemma}
\preservingorderlemma
\label{lem:presorder}
\end{lemma}
The above lemmas have implications for the smooth versions of 
Schur concave/convex functions. For any function $f$ from the space of
probability distribution to $\R$, let us define the following two smooth versions,
\begin{align}
\stipf(p)=\max_{\|q-p\|\leq \delta}f(q), \\
\flatf(p)=\min_{\|q-p\|\leq \delta}f(q).
\end{align}
Then, we have the following proposition, which allows for explicitly computing the 
smoothed entropies for a given value of $\delta$~\cite{deMeer-MSc,deMeer2017}.
\begin{proposition}
\label{prop:smoothed_schur}
Let $f$ be Schur-convex function, and $p$ a probability distribution of $k$ elements. Then 
\be
\stipf(p) = f(\stipp),\quad \flatf(p)=f(\flatp).
\label{eq:f-smooth-attain1}
\ee
If $f$ is a Schur-concave function, then
\be
\stipf(p) = f(\flatp),\quad \flatf(p)=f(\stipp).
\label{eq:f-smooth-attain2}
\ee
\end{proposition}
\begin{proof}
For $f$ being Schur-convex we have that $p\succ q$ implies 
\be
f(p) \geq f(q).
\ee
Thus the function preserves majorization order, hence on the set of $\delta$-approximations of $p$
it is maximal on $\stipp$ and minimal on $\flatp$.  Thus from definition of $\stipf$ and $\flatf$ we obtain 
Eq.~\eqref{eq:f-smooth-attain1}. An analogous argument applies when $f$ is Schur-concave.
\end{proof}
It directly follows from Lemma~\ref{lem:presorder} and Proposition~\ref{prop:smoothed_schur} that
\begin{corollary}
The smoothed versions of a Schur-convex function, $\stipf$ and $\flatf$, are monotonic under
majorization order, i.e., given two probability distributions of $k$ elements, $p$ and $q$, where
$p \succ q$, we have 
\be
\stipf(p) \geq \stipf(p),\quad \flatf(p)\geq \flatf(p).
\ee
\label{cor:monotone}
\end{corollary}
\par
We close this section with a result about the minimum distance $\delta$ which allows the
$\delta$-approximation of $p$ to majorize $q$ (and the $\delta$-approximation of $q$ to
be majorized by $p$), when $p \not\succ q$.
\def\minimumdeltaprop
{
Consider two probability distributions of $k$ elements, $p$ and $q$, such that $p \not\succ q$.
Let $\delta_1$ be the minimal $\delta$ such that $\stipp\succ q$, and $\delta_2$ the minimal
$\delta$ such that $p \succ \flatq$. Then we have
\be
\delta_1 = \delta_2= \delta^* \equiv
2 \max_{l \in \left\{1, \ldots, k \right\}} \sum_{i=1}^l \left( q_i - p_i \right).
\ee
}
\begin{proposition}
\minimumdeltaprop
\label{prop:mindelta}
\end{proposition}
The above proposition provides a measure of how much a distribution $p$ majorizes a distribution $q$,
in the sense that it tells us how much we have to distort $q$ in order for $p$ to majorize it, or equally
how much we have to distort $p$ in order for $q$ to be majorized by it. By proving the statement of
Proposition \ref{prop:mindelta} we have shown that this measure is given by $\delta^*$, which is
the minimal distance that allows the \stippest\ approximation of $p$ to majorize $q$, and the
\flattest\ approximation of $q$ to be majorized by $p$. Other measures of majorization distance
include the mixing character/distance~\cite{ruch1976principle,ruch1978mixing}, the information/work
distance~\cite{HO-limitations,faist_minimal_2015, brandao_second_2015}, and the maximum probability
of transition~\cite{alhambra_fluctuating_2016}.
\appendix
\section{Steepest and flattest approximations}
\begin{replem}{lem:flatstiplemma}
\label{replem:flatstiplemma}
\flatstiplemma
\end{replem}
\begin{proof}
Let $p$ be in non-increasing order. Let $\tilde p$ be an arbitrary $\delta$-approximation
of $p$, satisfying $\| \tilde p - p \| \leq \delta$. 
Then we can obtain $\tilde p$ as $\tilde p_i = p_i + \delta_i$, where $\sum_{i=1}^k
|\delta_i| \leq \delta$. Notice that the obtained probability distribution might be
not ordered, and therefore we define $\tilde p^{\downarrow}$ as the
non-increasingly ordered probability distribution obtained from $\tilde p$.
Also, notice that for all $m, l = 1 , \ldots , k$, with $m \leq l$, we have that
$\sum_{i=m}^l \delta_i \in \left[ - \frac{\delta}{2}, \frac{\delta}{2} \right]$.
We can now prove the Lemma.
\par
{\it Proof of part (i)}.
We will exploit the distribution $r$, see Eq.~\eqref{eq:nnormr}, used in the definition
of $\stipp$ (which is equal to $p$ with the largest element increased by $\frac{\delta}{2}$,
before the tail is cut by the same amount). Clearly, for any $l$ 
\be
\sum_{i=1}^l \tilde p^{\downarrow}_i \leq \sum_{i=1}^l p_i + \frac{\delta}{2} = \sum_{i=1}^l r_i.
\ee
Now, the procedure of cutting the tail only affects sums that are larger than $1$, and makes
them to be equal to $1$. 
Since in $\tilde p$ all sums are no greater than $1$, this does not affect the majorization conditions. 
Thus, we find that $\stipp$ majorizes $\tilde p$.
\par
{\it Proof of part (ii).}
Note that, over the interval $I$, the elements of $\flatp$ are all equal, see Eq.~\eqref{eq:flatp}.
The same is true for elements of $\flatp$ with indices in $J$. Those that are neither in $I$ nor
in $J$  are the same as in the original distribution $p$.  Since $p$ is in non-increasing order,
we have that $I=\left\{1,\ldots, l_I\right\}$, and $J=\left\{l_J,\ldots, k\right\}$.  Let us first consider sums up to $l$
elements for $l \leq l_I$. As we noticed, over the interval $I$ the distribution $\flatp$ is flat,
and its norm is equal to $\sum_{i=1}^{l_I} p_i - \frac{\delta}{2}$ due to the definition of $\ep(x^*)$,
see Eq.~\eqref{eq:epsx}. Let us now consider $\tilde p$, and its non-increasingly ordered version
$\tilde p^{\downarrow}$. Then it is clear that $\sum_{i=1}^{l_I} \tilde p^{\downarrow}_i \geq
\sum_{i=1}^{l_I} \tilde p_i \geq \sum_{i=1}^{l_I} p_i - \frac{\delta}{2}$, since the most we can
diminish the first $l_I$ largest elements of $p$ is by subtracting all $\frac{\delta}{2}$ from them. 
Therefore, we have two distributions over $I$, one is flat, and the second has larger total sum. 
Since all distributions majorize the flat one, we get that for all $l \leq  l_I$,
\be
\sum_{i=1}^l \tilde p^{\downarrow}_i \geq \sum_{i=1}^l \flatp_i.
\label{lem:majorpflatp}
\ee
For $l \not \in I \cup J$ we have that, according to its definition, $\flatp_i = p_i$, and therefore
the conditions of Eq.~\eqref{lem:majorpflatp} are still satisfied. To deal with the set of indices $J$, we
rewrite the related majorization inequalities which we need to prove as 
\be
\label{eq:flatpmajorinv}
\sum_{i=l}^k \flatp_i \geq \sum_{i=l}^k  \tilde p^{\downarrow}_i, \quad \forall \, l > l_J.
\ee
As a first step, we consider these sums for $l = l_J$. We have that,
\be
\sum_{i=l_J}^k  \tilde p^{\downarrow}_i \leq \sum_{i=l_J}^k  \tilde p_i = \sum_{i=l_J}^k  \left( p_i +\delta_i \right)
\leq \sum_{i=l_J}^k  p_i + \frac{\delta}{2} = \sum_{i=l_J}^k  \flatp_i,
\ee
where the last equality follows from the definitions of $\flatp$ and $\gamma(y^*)$, see Eq.~\eqref{eq:gammax}.
To prove Eq.~\eqref{eq:flatpmajorinv}, note that for $l \geq l_J$ we have  $\flatp_i = y^*$, where
$y^*$ is some positive number defined in the course of the construction of $\flatp$, see Eq.~\eqref{eq:flatp}.
For now, the value of $y^*$ is not important, and what we need is that all $\flatp_i$ are constant for
$i\in J$. Then, we have to prove that 
\be
\sum_{i=l}^k \tilde p^{\downarrow}_i \leq (k + 1 - l) \, y^*,
\ee
for $l > l_J$. But this is a consequence of the following easy-to-prove observation.
Consider non-negative numbers $\{a_i\}_{i=1}^n$ put in increasing order. Let
$\sum_{i=1}^n  a_i \leq n \, \lambda$, where $\lambda \geq 0 $ is some constant.
Then $\sum_{i=1}^l a_i\leq l \, \lambda$. This observation ends the proof.
\end{proof}
\begin{replem}{lem:presorder}
\label{replem:preservingorderlemma}
\preservingorderlemma
\end{replem}
\begin{proof}
Let us first prove that the {\it \stippest\ $\delta$-approximation preserves the majorization order.}
Let $l^*$ and $m^*$ be the indices defined as in Eq.~\eqref{eq:deflstar} for, respectively, $\stipp$
and $\stipq$. Note that $l^*\leq m^*$. Indeed, from the construction of $\stipp$ it follows that $l^*$ 
is the largest $l$ such that $\sum_{i=1}^{l} p_i + \frac{\delta}{2} \leq 1$. Then, from the fact that
$p \succ q$, we get that $l^* \leq m^*$. 
Now, for $l \leq l^*$ we have 
\be
\sum_{i=1}^l \stipp_i = \sum_{i=1}^l p_i+ \frac{\delta}{2}, \quad \sum_{i=1}^l \stipq_i=
\sum_{i=1}^l q_i+ \frac{\delta}{2},
\ee
hence $p\succ q$ implies
\be
\sum_{i=1}^l \stipp_i \geq \sum_{i=1}^l \stipq_i, \quad \forall \, l \leq l^*.
\ee
For $l > l^*$, instead, we have $\sum_{i=1}^l \stipp_i = 1$, hence the rest of the majorization conditions
is automatically satisfied. 
\par
Now we will prove that the {\it \flattest\ approximation preserves the majorization order.}
Following the definition of Eq.~\eqref{eq:flatp}, let us denote $x = x^*(p)$, $x' = x^*(q)$
and $y = y^*(p)$, $y' = y^*(q)$, where $x^*$, $y^*$ are defined in the course of constructing
the \flattest\ approximation; $x^*$ is the level at which the first largest elements are cut, 
and $y^*$ is the level to which the smallest elements are enlarged. 
Similarly, let us denote $l_I \equiv l_I(p)$, $l_I' \equiv l_I(q)$ and $l_J \equiv l_J(p)$, $l_J' \equiv l_J(q)$.
Recall that the interval $I = \left\{1,\ldots, l_I \right\}$ labels the elements that are cut (and
have become equal to $x^*$), while the interval $J = \left\{ l_J, \ldots, k \right\}$ labels the elements
that are enlarged (and have become equal to $y^*$).
In the following, we will frequently use the result of Lemma \ref{leqm:xy}, that $x \geq x'$, and
$y \leq y'$. In fact, these inequalities are necessary conditions for majorization (indeed, $x$
and $x'$ are the largest elements, while $y$ and $y'$ are the smallest elements of $\flatp$ and $\flatq$,
respectively).
\par
We will divide the range $l = 1, \ldots, k$ into five intervals; (i) $\left[1, l_I \right]$,
(ii) $\left[ l_I+1, l_I' \right]$, (iii) $\left[ l_J, k \right]$, (iv) $\left[ l_J', l_J-1 \right]$, and
(v) $\left( \max \left\{ l_I, l_I' \right\}, \min \left\{ l_J, l_J' \right\} \right)$. Notice that the
intervals (ii) and (iv) may be empty. For each interval we will prove that
\be
\sum_{i=1}^{l} \flatp_i \geq \sum_{i=1}^{l} \flatq_i,
\label{eq:major_approx}
\ee 
for $l$ belonging to the specific interval.
\bei
\item[(i)] $[1,l_I]$: This case is immediate. For all $ i \leq l_I$, and independently of whether
$l_I > l_I'$ or vice versa, we have
\be
\flatp_i = x \geq x' \geq \flatq_i,
\ee
where the first inequality follows from Lemma \ref{leqm:xy}, and the second one from the
fact that $x'$ is largest element of $\flatq$. Notice that the second inequality is saturated
for all $i \leq  l_I'$. Summing up we obtain Eq.~\eqref{eq:major_approx} for $l\leq l_I$. 
\item[(ii)]$[l_I+1, l_I']$: This case is trivial if the set is empty. When the set is not empty,
instead, we start by considering the case of $l = l_I'$. In this situation we have
\be
\sum_{i=1}^{l_I'} \flatq_i= \sum_{i=1}^{l_I'} q_i - \frac{\delta}{2}, \quad 
\sum_{i=1}^{l_I'} \flatp_i \geq \sum_{i =1}^{l_I'} p_i - \frac{\delta}{2},
\ee 
which follows from the definition of $\ep(x)$ and $\ep(x')$, see Eq.~\eqref{eq:epsx}, and the fact
that $l_I \leq l_I'$. Thus, due to the fact that $p\succ q$, we obtain 
\be
\sum_{i=1}^{l_I'} \flatp_i  \geq \sum_{i=1}^{l_I'} \flatq_i.
\ee
Then, since $\flatp$ has no smaller norm than $\flatq$ on this interval, and moreover $\flatq$
is flat on the interval, we have that $\flatp$ (as well as any other distribution with no smaller
norm) majorizes $\flatq$ on the interval. This proves Eq.~\eqref{eq:major_approx} for this
interval.
\item[(iii)] $[l_J,k]$: 
In this interval we will prove equivalent relation to the one of Eq.~\eqref{eq:major_approx},
namely
\be
\sum_{i=l}^k \flatp_i\leq \sum_{i=l}^k \flatq_i,
\label{eq:major_approx2}
\ee
for $l > l_J$.  For all $i \geq l_J$, and independently of whether $l_J < l_J'$ or vice versa,
we have 
\be
\flatp_i = y \leq y' \leq \flatq_i,
\ee
where the first inequality follows from Lemma \ref{leqm:xy}, and the second one from the
fact that $y'$ is the smallest element of $\flatq$. Summing up, we obtain Eq.~\eqref{eq:major_approx2}
for $l \geq l_J$. 
\item[(iv)] $[l_J',l_J-1]$: This case is trivial if the set is empty. When the set is not empty,
instead, we start by considering the case of $l = l_J'$. We have that
\be
\sum_{i=l_J'}^{k} \flatq_i= \sum_{i=l_J'}^{k} q_i + \frac{\delta}{2}, \quad 
\sum_{i=l_J'}^{k} \flatp_i \leq \sum_{i =l_J'}^{k} p_i + \frac{\delta}{2},
\ee
which follows from the definition of $\gamma(y)$ and $\gamma(y')$, see Eq.~\eqref{eq:gammax},
and the fact that $l_J \geq l_J'$. Therefore, by $p \succ q$ we obtain that
\be
\sum_{i=l_J'}^{k} \flatp_i \leq \sum_{i=l_J'}^{k} \flatq_i.
\ee
Thus, on this interval $\flatq$ has no smaller norm than $\flatp$, and moreover $\flatq$ is flat. 
If the norms were equal to each other, $\flatp$ would majorize $\flatq$ on the interval,
since any distribution majorizes the flat distribution. Therefore the conditions 
\be
\sum_{i=l}^k \flatp_i \leq \sum_{i=l}^k \flatq_i,
\ee
would be satisfied for $l > l_J'$.  Since norm of $\flatq$ may only be larger, the above inequalities still hold. 
\item[(v)] $\max\{l_I,l_I'\} < l < \min\{l_J,l_J'\}$: Note that for $l$ in such interval we have 
\be
\sum_{i=1}^l \flatp_i =  \sum_{i=1}^l p_i - \frac{\delta}{2},
\quad \sum_{i=1}^l \flatq_i = \sum_{i=1}^l q_i - \frac{\delta}{2},
\ee
which follows from the definition of $\ep(x)$ and $\ep(x')$, see Eq.~\eqref{eq:epsx}.
So, by $p \succ q$ we obtain for the considered interval
\be
\sum_{i=1}^l \flatp_i \geq \sum_{i=1}^l \flatq_i.
\ee
This concludes the proof of the majorization relations for all $l$. 
\eei
\end{proof}
\begin{lemma}
\label{leqm:xy}
Let us consider two probability distributions of $k$ elements, $p$ and $q$, where
$p \succ q$. We denote $x= x^*(p)$, $x'=x^*(q)$ and $y= y^*(p)$, $y'=y^*(q)$, 
where $x^*$, $y^*$ are defined in the course of constructing the \flattest\ approximation;
$x^*$ is the level at which first largest elements are cut, and $y^*$ is the level to which the
smallest elements are enlarged. Then 
\be
x \geq x',\quad y \leq y'.
\ee
\end{lemma}
\begin{proof}
Let us first denote $l_I \equiv l_I(p)$, $l_I' \equiv l_I(q)$, and $l_J \equiv l_J(p)$, $l_J' \equiv l_J(q)$.
Recall that the interval $I = \left\{1,\ldots, l_I \right\}$ labels the elements that are cut (and
have become equal to $x^*$), while the interval $J = \left\{ l_J, \ldots, k \right\}$ labels the elements
that are enlarged (and have become equal to $y^*$).
\par
To prove that $x \geq x'$, notice first that for any $l=1,\ldots, k$ we have 
\be
\sum_{i=1}^l p_i \leq l \, x + \frac{\delta}{2},
\label{eq:p_vs_x}
\ee
Indeed, for $l\leq l_I$ we have 
\be
\sum_{i=1}^l p_i= \sum_{i=1}^{l} (x + \delta_i) \leq  l \, x + \frac{\delta}{2},
\ee
where $\delta_i\equiv p_i - \flatp_i$ satisfy $\sum_{i=1}^{l_I}\delta_i=\frac{\delta}{2}$ 
and $\sum_{i=1}^l \delta_i \geq 0$ for any $l$, which follows from the construction of $\flatp$. 
For $l > l_I$, instead, we have
\be
\sum_{i=1}^l p_i = \sum_{i=1}^{l_I} p_i + \sum_{i=l_I + 1}^l p_i = l_I \, x + \frac{\delta}{2} + \sum_{i=l_I+1}^l p_i
\leq  l_I \, x + \frac{\delta}{2} + (l - l_I) \, x = l \, x + \frac{\delta}{2},
\ee
where the inequality follows from the definition of the interval $I$, see Eq.~\eqref{eq:setI}.
Now we use Eq.~\eqref{eq:p_vs_x} for $l = l_I'$, in conjunction with the majorization condition
$p \succ q$, to get $x \geq x'$. We write
\be
l_I' \, x + \frac{\delta}{2} \geq \sum_{i=1}^{l_I'} p_i \geq \sum_{i=1}^{l_I'} q_i = l_I' \, x' + \frac{\delta}{2},
\ee
which implies $x \geq x'$, since  $l_I'\geq 1$ by definition. 
\par
The relation $y \leq y'$ is proved in an analogous way. First, for any $l = 1,\ldots,k$ we have 
\be
\sum_{i=l}^k p_i \geq (k-l+1) \, y - \frac{\delta}{2}.
\label{eq:q_vs_y}
\ee
Indeed, for $l \geq l_J$ we have 
\be
\sum_{i=l}^k p_i = \sum_{i=l}^k \left( y - \ep_i \right) \geq (k-l+1) \, y - \frac{\delta}{2},
\ee
where $\ep_i\equiv \flatp_i -p_i$ satisfy $\sum_{i=l_J}^k\ep_i=\frac{\delta}{2}$ and $\sum_{i=l}^k \ep_i\geq 0$ 
for any $l$, which follows from the construction of $\flatp$. 
For $l<l_J$, instead, we have 
\be
\sum_{i=l}^{k} p_i = \sum_{i=l}^{l_J-1} p_i + \sum_{i=l_J}^{k} p_i =
\sum_{i=l}^{l_J-1} p_i + (k- l_J+1) \, y - \frac{\delta}{2} \geq (k-l+1) \, y - \frac{\delta}{2},
\ee
where the inequality follows from the definition of the interval $J$, see Eq.~\eqref{eq:setJ}.
Now we use Eq.~\eqref{eq:q_vs_y} for $l = l_J'$, in conjunction with the majorization
condition $p \succ q$, to show that $y \leq y'$. We write
\be
(k-l_J'+1) \, y  - \frac{\delta}{2} \leq \sum_{i=l_J'}^k p_i  \leq \sum_{i=l_J'}^k q_i = (k-l_J'+1) \, y'  - \frac{\delta}{2},
\ee
which implies $y\leq y'$.
\end{proof}
\begin{repprop}{prop:mindelta}
\label{repprop:mindelta}
\minimumdeltaprop
\end{repprop}
\begin{proof}
Let us begin by showing that $\delta^*$ is the minimum distance $\delta$ such that $\stipp \succ q$.
As a first step, we want to show that $\stippstar$ majorizes $q$. To this aim, consider the non-normalised
distribution $r$ obtained from $p$ by adding $\frac{\delta^*}{2}$ to its first element, Eq.~\eqref{eq:nnormr}.
Then, for all $l \leq l^*$, we have
\be
\sum_{i=1}^{l} \stippstar_i = \sum_{i=1}^{l} r_i = \sum_{i=1}^{l} p_i + \frac{\delta^*}{2} \geq
\sum_{i=1}^{l} p_i + \left( \sum_{i=1}^{l} \left( q_i - p_i \right) \right) = \sum_{i=1}^{l} q_i,
\ee
where the inequality follows from the definition of $\delta^*$. When $l > l^*$, instead, we
have that $\sum_{i=1}^{l} \stippstar_i = 1$, and due to the normalisation condition on $q$ we
have that $\sum_{i=1}^{l} \stippstar_i \geq \sum_{i=1}^{l} q_i$. Then, $\stippstar \succ q$.
\par
To show that $\delta^*$ is minimum, we consider $\bar{\delta} < \delta^*$, and we show that
$\stippbar \not \succ q$. In this case, it exists an $\bar{l}$ such that
\be
\label{eq:deltabar}
\frac{\bar{\delta}}{2} < \sum_{i=1}^{\bar{l}} \left( q_i - p_i \right).
\ee
Then,
\be
\sum_{i=1}^{\bar{l}} \stippbar_i \leq \sum_{i=1}^{\bar{l}} p_i + \frac{\bar{\delta}}{2} <
\sum_{i=1}^{\bar{l}} p_i + \sum_{i=1}^{\bar{l}} \left( q_i - p_i \right) =  \sum_{i=1}^{\bar{l}} q_i,
\ee
where the first inequality is saturated for $\bar{l} \leq l^*$, and the second inequality follows
from Eq.~\eqref{eq:deltabar}. Thus, we have that $\stippbar \not \succ q$ for all $\bar{\delta}
< \delta^*$.
\par
Now, we show that $\delta^*$ is the distance $\delta$ such that $p \succ \flatq$.
In particular, we initially want to show that $p \succ \flatqstar$. As a first step, we consider
the interval $I = \left\{1, \ldots , l_I \right\}$ in which $\flatqstar$ is flat, and all its elements are
equal to $x^*$. In particular, we have that
\be
\sum_{i=1}^{l_I} \flatqstar_i = \sum_{i=1}^{l_I} x^*_i = \sum_{i=1}^{l_I} q_i - \frac{\delta^*}{2}
\leq \sum_{i=1}^{l_I} q_i - \left( \sum_{i=1}^{l_I} \left( q_i - p_i \right) \right) = \sum_{i=1}^{l_I} p_i,
\ee
where the second equality directly follows from Eq.~\eqref{eq:epsx} and from the fact that
$\ep(x^*) = \frac{\delta^*}{2}$, while the inequality follows from the definition of $\delta^*$.
The above equation proves that, on the interval $I$, the norm of $\flatqstar$ is smaller or
equal to the one of $p$. Then, since $\flatqstar$ is flat over the interval $I$, we have that
$p$ majorizes it, that is,
\be
\label{eq:flatpstar1}
\sum_{i=1}^l p_i \geq \sum_{i=1}^l \flatqstar_i, \quad \forall \, l \leq l_I.
\ee
We can now consider the interval in between $I$ and $J$, where $J = \left\{ l_J , \ldots, k \right\}$.
For all $l$ in this interval, $l_I < l < l_J$, we have
\be
\sum_{i=1}^{l} \flatqstar_i = \sum_{i=1}^{l_I} x^* + \sum_{i=l_I+1}^{l} q_i = \sum_{i=1}^{l} q_i - \frac{\delta^*}{2}
\leq \sum_{i=1}^{l} q_i - \left( \sum_{i=1}^{l} \left( q_i - p_i \right) \right) = \sum_{i=1}^{l} p_i,
\ee
which, again, follows from the definition of $\ep(x^*)$ and the one of $\delta^*$. Thus, we find that
\be
\label{eq:flatpstar2}
\sum_{i=1}^l p_i \geq \sum_{i=1}^l \flatqstar_i, \quad \forall \, l \in \left\{ l_I + 1 , \ldots, l_J - 1 \right\}.
\ee
Finally, we consider the interval $J$. In this case, we will prove that
\be
\label{eq:flatpstar3}
\sum_{i=l}^k p_i \leq \sum_{i=l}^k \flatqstar_i, \quad \forall \, l > l_J.
\ee
To do so, let us consider the case $l = l_J$, where we have
\be
\sum_{i=l_J}^k \flatqstar_i = \sum_{i=l_J}^k y^* = \sum_{i=l_J}^{k} q_i + \frac{\delta^*}{2} \geq
\sum_{i=l_J}^{k} q_i + \left( \sum_{i=l_J}^{k} \left( p_i - q_i \right) \right) = \sum_{i=l_J}^{k} p_i,
\ee
which follows from the definition of $\gamma(y^*)$ and the one of $\delta^*$. Thus, we have that,
over the interval $J$, $\flatqstar$ has bigger norm than $p$. Then, following the same argument used
in the proof of Lemma~\ref{lem:presorder} (iv), we have that since $\flatqstar$ is flat over $J$, then it is majorized by
$p$, which proves Eq.~\eqref{eq:flatpstar3}. Therefore, we have that $p \succ \flatqstar$.
\par
To conclude the proof, we need to show that $\delta^*$ is minimum, that is, for all $\bar{\delta} < \delta^*$,
we have that $p \not \succ \flatpbar$. When $\bar{\delta}$ is considered, we have seen that an $\bar{l}$
exists such that Eq.~\eqref{eq:deltabar} is satisfied. Then, for $l = \bar{l}$,
\be
\sum_{i=1}^{\bar{l}} \flatpbar_i \geq \sum_{i=1}^{\bar{l}} q_i - \frac{\bar{\delta}}{2} >
\sum_{i=1}^{\bar{l}} q_i - \sum_{i=1}^{\bar{l}} \left( q_i - p_i \right) =  \sum_{i=1}^{\bar{l}} p_i,
\ee
where the first inequality is saturated when $l_I \leq \bar{l} < l_J$, and the second one follows
from Eq.~\eqref{eq:deltabar}. Thus, we have that $p \not \succ \flatpbar$ for all $\bar{\delta}
< \delta^*$.
\end{proof}
{\bf Acknowledgements}
We thank Fernando Brand\~ao, Nelly Ng and Stephanie Wehner for discussions.
MH is partially supported by a grant from the John Templeton Foundation.
The opinions expressed in this publication are those of the
authors and do not necessarily reflect the views of the John Templeton Foundation.
JO thanks the Royal Society and an EPSRC Established Career Fellowship for their support.
CS is supported by the EPSRC [grant number EP/L015242/1].
\bibliography{biblio_smooth}
\end{document}

%% file: definitions.tex
\def\<{\langle}
\def\>{\rangle}

\def\stipp{\bar{p}^{(\delta)}}
\def\stippstar{\bar{p}^{(\delta^*)}}
\def\stippbar{\bar{p}^{(\bar{\delta})}}
\def\stipq{\bar{q}^{(\delta)}}
\def\stipf{\bar{f}^{(\delta)}}
\def\flatp{\underline{p}^{(\delta)}}
\def\flatq{\underline{q}^{(\delta)}}
\def\flatqstar{\underline{q}^{(\delta^*)}}
\def\flatpbar{\underline{q}^{(\bar{\delta})}}
\def\flatf{\underline{f}^{(\delta)}}
\def\flattest{flattest}
\def\stippest{steepest}

\newcommand{\be}{\begin{eqnarray} \begin{aligned}}
\newcommand{\ee}{\end{aligned} \end{eqnarray} }
\newcommand{\benn}{\begin{eqnarray*} \begin{aligned}}
\newcommand{\eenn}{\end{aligned} \end{eqnarray*} }

\newcommand{\ben}{\begin{eqnarray} \begin{aligned}}
\newcommand{\een}{\end{aligned} \end{eqnarray} }

\newcommand{\bc}{\begin{center}}
\newcommand{\ec}{\end{center}}

\newcommand{\e}{\mathrm{e}}

\newcommand{\beq}{\begin{eqnarray} \begin{aligned}}
\newcommand{\eeq}{\end{aligned} \end{eqnarray} }
\newcommand{\bea}{\begin{array}}
\newcommand{\eea}{\end{array}}

\newcommand{\bee}{\begin{enumerate}}
\newcommand{\eee}{\end{enumerate}}
\newcommand{\bei}{\begin{itemize}}
\newcommand{\eei}{\end{itemize}}
\newtheorem{theorem}{Theorem}
\newtheorem{proposition}[theorem]{Proposition}

\newtheorem{lemma}[theorem]{Lemma}

\newtheorem{corollary}[theorem]{Corollary}

\usepackage{amsfonts}

\def\01{\{0,1\}}

\def\<{\langle}
\def\>{\rangle}

\def\ep{\epsilon}

\newtheorem*{rep@theorem}{\rep@title}
\newcommand{\newreptheorem}[2]{%
\newenvironment{rep#1}[1]{%
 \def\rep@title{#2 \ref{##1} (restatement)}%
 \begin{rep@theorem}}%
 {\end{rep@theorem}}}
\makeatother
\newreptheorem{thm}{Theorem}
\newreptheorem{lem}{Lemma}
\newreptheorem{prop}{Proposition}

\newcommand{\R}{\mathbb{R}}